\documentclass[10pt]{article}
\usepackage[utf8]{inputenc}
\usepackage{authblk}

\usepackage{graphicx}
\usepackage{listings}

\usepackage{xcolor}
\usepackage[spaces,hyphens]{url}
\usepackage{hyperref}
\usepackage{breakurl}

\usepackage[strings]{underscore}

\usepackage{subcaption}
\usepackage{float}
\newfloat{algorithm}{t}{lop}

\makeatletter
\def\IState{\State\hspace{10pt}}
\def\IIState{\State\hspace{20pt}}
\def\IIIState{\State\hspace{30pt}}

\makeatother

\usepackage{comment}

\usepackage{xspace}
\usepackage[ruled]{algorithm}
\usepackage[noend]{algpseudocode}
\usepackage{multicol}
\usepackage[cmex10]{amsmath}
\usepackage{amsthm}
\newtheorem{theorem}{Theorem}
\newtheorem{definition}{Definition}[section]
\newtheorem{lemm}{Lemma}
\usepackage{amssymb,amsfonts}
\usepackage{textcomp}

\usepackage{cleveref} 
\usepackage{hyperref} 

\let\olditemize\itemize
\renewcommand\itemize{\olditemize\setlength\itemsep{1pt}\parskip0pt}
\let\oldenumerate\enumerate
\renewcommand\enumerate{\oldenumerate\setlength\itemsep{1pt}\parskip0pt}

\newcommand*\samethanks[1][\value{footnote}]{\footnotemark[#1]}

\usepackage[backend=bibtex, style=numeric, sorting=none]{biblatex}
\begin{document}
	
	\title{AToM: Active Topology Monitoring for the Bitcoin Peer-to-Peer Network
	}
	
		
	
	\author[]{Federico Franzoni\thanks{These authors were supported by Project RTI2018-102112-B-I00 (AEI/FEDER,UE).}} 
	\author[]{Xavier Salleras}
	\author[]{Vanesa Daza\samethanks[1]} %
	\affil[]{Universitat Pompeu Fabra, Barcelona, Spain}
	\affil[ ]{\texttt{\{federico.franzoni,xavier.salleras,vanesa.daza\}@upf.edu}}

	\date{}
	
	\maketitle
	
	\begin{abstract}
	Over the past decade, the Bitcoin P2P network protocol has become a reference model for all modern cryptocurrencies.
	While nodes in this network are known, the connections among them are kept hidden, as it is commonly believed that this helps protect from deanonymization and low-level attacks.
	However, adversaries can bypass this limitation by inferring connections through side channels.
	At the same time, the lack of topology information hinders the analysis of the network, which is essential to improve efficiency and security.
	In this paper, we thoroughly review network-level attacks and empirically show that topology obfuscation is not an effective countermeasure. 
	We then argue that the benefits of an open topology potentially outweigh its risks, and propose a protocol to reliably infer and monitor connections among reachable nodes of the Bitcoin network.
	We formally analyze our protocol and experimentally evaluate its accuracy in both trusted and untrusted settings.
	Results show our system has a low impact on the network, and has
	precision and recall are over 90\% with up to 20\% of malicious nodes in the network.
	\end{abstract}

\section{Introduction}
\label{sec:introduction}
Since its release, Bitcoin~\cite{nakamoto2008bitcoin} has attracted a constantly increasing number of users, who exchange large amounts of money every day~\cite{coinmarketcap}.
Given its relevance, the security aspects of this cryptocurrency system have been extensively studied in research~\cite{conti2018survey,tapsell2018evaluation,reid2013analysis}.
The network layer, among other components, has received a lot of attention due to its key role in the communication between the involved actors. 
In particular, several practical attacks have been shown, such as  partitioning~\cite{neudecker2015simulation}, eclipse~\cite{heilman2015eclipse}, and deanonymization~\cite{biryukov2014deanonymisation,koshy2014analysis}, which can pose a serious threat to the 
security of single users or that the system as a whole.

With respect to such attacks, Bitcoin developers implicitly adopted the concealment of the network topology as a protective measure.
Specifically, while (reachable) nodes in the network are publicly known~\cite{bitnodes}, their connections are kept hidden.
This security-by-obscurity approach is meant to hinder adversaries from performing attacks like the ones we mentioned above.
However, 
several techniques have been shown over the years that allow bypassing this measure by exploiting side channels~\cite{miller2015discovering,neudecker2016timing} or abusing the protocol~\cite{delgado2018txprobe}.
While Bitcoin developers promptly fix the protocol to make newly-disclosed techniques ineffective, it is hard to prevent attackers from using undisclosed methods or devising new ones.
At the same time, hiding the topology hinders a proper analysis of the network~\cite{miller2015discovering,delgado2018txprobe} and an accurate definition of network models~\cite{neudecker2015simulation,neudecker2019characterization}.
As a consequence, these limitations ultimately hinder the improvement of the efficiency and security of the network.
In contrast, having a reliable source of topology information could enable the design of a safer and more performing network.

It is thus important to investigate whether concealing the topology is indeed a valid protection mechanism for the Bitcoin network, given the fact that topology information is not a threat per se and that no solid proof has been given to support the current approach.
In this paper, we empirically show the ineffectiveness of topology obfuscation as a defense from known network-level attacks, and foster the idea of a public topology.
In light of this, we argue that the benefits of an open topology potentially outweigh its risks, 
and propose a protocol to reliably infer and monitor connections among reachable nodes of the network.

We analyze the potential impact of our protocol on the Bitcoin network and experimentally evaluate its performance in a simulated environment, as well as its resilience to malicious nodes.
Results show that our system has little impact on the network and has high precision and recall, even with high concentrations of malicious nodes.
Furthermore, although designed for Bitcoin, our solution can be implemented on any P2P network.

The contributions of our paper include:
\vspace{-5pt}
\begin{itemize}
	\item we study the potential benefits of an open network topology for the Bitcoin P2P network;
	\item we show that most attacks mentioned in literature as aided by topology information are actually independent from such knowledge;
	\item we design a protocol to prove connections among peers, and propose a system to monitor the topology of the reachable network; 
	\item we employ a reputation scheme to discourage malicious nodes from misbehaving and, at the same time, decentralize trust in the monitoring nodes;
	\item we implement a proof of concept and evaluate the accuracy of our solution through simulations.
\end{itemize}

\section{The Bitcoin P2P Network}
\label{sec:background}
The Bitcoin P2P network is the core infrastructure through which clients exchange transactions and reach consensus on the contents of the blockchain.
New nodes join the network by connecting to well-known nodes or querying a trusted DNS server.
After establishing the first connections, nodes discover new peers by receiving \texttt{ADDR} messages from their current neighbors. 
These messages advertise other known nodes in the network to which it is possible to connect.

Each node connects to some peers and, when possible, accepts connections from other nodes.
Relatively to a node, connections are called \textit{outbound} if initiated by the node itself, and \textit{inbound} otherwise.
Nodes running the reference client (Bitcoin Core)~\cite{bitcoincore}, which is the most used in the network~\cite{bitnodes}, always keep 8 outbound connections, and accept inbound connections up to a maximum of 125 peers~\cite{bitcoinprotocol}. 
Nevertheless, this limitation is not enforced by the protocol, leaving nodes free to establish any number of connections.

As some clients, due to the presence of NAT or firewall, are unable to accept incoming connections, nodes in the network are typically divided into \textit{reachable} and \textit{unreachable}~\cite{delgado2018cryptocurrency}.
According to research, unreachable nodes are around ten times more than reachable nodes, thus counting for almost 90\% of the whole network~\cite{wang2017towards,delgado2018txprobe}.
Despite that, most research focuses exclusively on the reachable part of the network~\cite{donet2014bitcoin,miller2015discovering}, as this is easier to analyze, and is considered to be more relevant for the propagation of messages since that they maintain the vast majority of connections~\cite{decker2013information}.
Similarly, in this work, we will only focus on reachable nodes.

\section{Motivation}
\label{sec:motivation}
The relevance of topology for blockchain networks has been clearly shown in the literature.

Kiayias et al.\cite{kiayias2015speed} show how the efficiency of the information propagation is directly influenced by the network topology.
Propagation delay can also be substantially affected by the number of connections that nodes have~\cite{essaid2018network}.
Furthermore, all unstructured P2P networks, like blockchain ones, are known to suffer from the so-called \textit{topology mismatch} problem (i.e., the incoherence between logical and physical links), which causes inefficiency in the transmission of data~\cite{zhao2006solving,liu2007building}.

As for what concerns security, it has been shown how the Bitcoin network topology is far from being a random graph as designed.
Instead, it contains nodes communities~\cite{delgado2018txprobe} and supernode~\cite{miller2015discovering}, showing high levels of centralization.
Knowledge of the network topology could help spot such issues and improve decentralization.
An open topology could also boost the propagation of transactions and blocks, which in turn would reduce the ability to perform double-spending attacks and selfish mining \cite{dotan2020sok}.

If nodes had access to topology information, all these aspects could be addressed in real-time.
Furthermore, the availability of such information could help reduce discrepancies between formal models and the real network~\cite{kiayias2015speed}.

As stated by Delgado et al.~\cite{delgado2018txprobe}, hiding the topology prevents network analysis and measurement, further complicating the solution of existing issues.
Similarly, Miller et al.~\cite{miller2015discovering} point out that understanding the topology allows identifying structural faults in the network that might hinder broadcast optimization.
For this reason, they support the idea that monitoring the network can help quickly detect and react to attacks and mistakes.
Authors of \cite{deshpande2018btcmap} also state it is crucial to acquire the knowledge of topology to accurately manage the network, optimize its performances, and ensure that it works properly.

An open Bitcoin topology could allow the introduction of mechanisms for avoiding centralization and increasing connectivity among nodes, as well as detecting weak spots that can be exploited to perform network-level attacks.
With this work, we aim at taking a first step towards this direction.
We do this by showing that an open topology should not be considered a security concern, and by proposing a practical protocol to reliably monitor the state of the network.

\section{Security Concerns of Open Topology}
\label{sec:topologysecurity}
According to research, the main reason to keep the topology hidden is to avoid the risk of network-level attacks and deanonymization~\cite{miller2015discovering,grundmann2019exploiting,delgado2018txprobe}.
In this section, we analyze known threats and evaluate their relation to topology knowledge.

\subsection{Network-level Attacks}
Network attacks commonly related to topology information include double-spending, selfish mining, partitioning, and eclipse attacks.
In the following, we study each of these attacks. 

\paragraph{\textbf{Double Spending Attacks}}
Double spending in fast payment transactions~\cite{karame2012doublespending} was one of the first attacks correlated with topology. 
In this attack, the adversary sends a transaction to a victim node, while sending a conflicting one to the rest of the network.
The goal is to make the victim accept a 0-confirmation transaction as payment, while having a double-spending one added to the blockchain.
Note that, despite being cited by many research papers, the attack described in \cite{karame2012doublespending} does not make use of any topological information beyond the IP address of the victim, which is publicly available.
Moreover, as the probability of success decreases exponentially in the number of peers, the victim could effectively protect himself by establishing few more connections.

Knowledge of the topology might indeed ease a double-spending attack when the victim is only connected to few, reachable, peers~\cite{lei2015exploiting}.
However, this situation is typical of unreachable nodes, which would not be affected by revealing the topology of reachable nodes.
Furthermore, the time frame in which this attack can succeed is so short that the victim only needs to wait a couple of seconds to be safe.

\paragraph{\textbf{Mining Attacks}}
Mining attacks~\cite{eyal2014majority,nayak2016stubborn} could also benefit from the knowledge of the topology.
In particular, miners could take advantage of network information to improve the propagation of their own blocks, or to slow down competing ones~\cite{miller2015discovering,nayak2016stubborn}.
However, this is true only when a fraction of miners have such knowledge.
Instead, if all miners had access to topology information, they could all leverage highly connected nodes to speed up block propagation, with no advantage for a particular subset.

On the other hand, the drastic improvement in block propagation speed \cite{dsn-block-propagation} likely made topology information less relevant for this kind of attacks.
Moreover, nowadays, miners usually connect to each other through high-speed relay networks~\cite{bitcoinfibre,falconnet}, which are separate from the main network, thus further reducing the importance of its topology. 

\paragraph{\textbf{Partitioning Attacks}}
Partitioning attacks~\cite{neudecker2015simulation,saad2019partitioning} are also a commonly cited problem in relation to topology knowledge.
These attacks consist in trying to split the network to prevent communication between isolated groups of nodes.
The goal of the attacker can be to double-spend or simply to affect the functionality of the network and create distrust in the system.
Depending on the goal, the attacker can follow different strategies.
To divide the network in two she tries to detect the minimum vertex cut, that is, the smallest set of nodes whose removal causes a split in the network, and disrupt their communications with a DoS attack.
To disconnect a specific portion of the network, she tries to disrupt all the connections between the target nodes and the rest of the network. 

In both situations, knowing the topology would indeed make the attack easier.
However, the reference study on the topic \cite{neudecker2015simulation} shows how the network can resist attacks lasting several hours even against a powerful adversary controlling a botnet as large as the Bitcoin network itself.
Moreover, for the attack to be effective, both reachable and unreachable nodes have to be considered.
For instance, let us consider a public topology containing a cut $C$ between portions $P1$ and $P2$ of the network. 
By looking at it the attacker might think that taking down $C$, $P1$ and $P2$ would be disconnected (i.e., partitioned); however, the attacker does not see any of the unreachable nodes that might be connected to both $P1$ and $P2$.
Therefore, knowing the topology of the reachable network is not sufficient to calculate a complete cut.

Note that actively monitoring the network graph would also provide a double benefit for protecting from these attacks.
On the one hand, it would help detect attacks in real-time, allowing the network to react promptly.
On the other hand, using an adaptive topology-aware protocol, it would be possible to maximize the number of nodes in the minimum vertex cut (of the reachable network) to make the attack harder.

Finally, it is worth mentioning that, as pointed out in~\cite{delgado2018cryptocurrency}, partitioning can accidentally occur in random P2P networks if the network formation process is not well designed.
This is one of the reasons why Bitcoin makes use of trusted DNS servers and hard-coded peers for the bootstrap procedure.
Again, it is easy to show that monitoring the public topology could help prevent such an event.

\paragraph{\textbf{Routing Attacks}}
Similarly to partitioning attacks, routing attacks \cite{apostolaki2017hijacking,tran2020stealthier} aim at 
isolating a portion of nodes at the AS-level by means of BGP hijacking.
These attacks are also commonly cited to justify the obfuscation of the topology.

However, in these attacks, the adversary gains advantage from being in control of an AS, and only needs knowledge of the target IP addresses, with no need of further topology information.
Furthermore, an AS-level adversary already has knowledge of all connections under its control, which represents a potential advantage over network-level nodes.
Making the topology public to the whole network would effectively reduce this advantage.

\paragraph{\textbf{Eclipse Attacks}}
Another attack typically mentioned in relation to network topology is the eclipse attack~\cite{heilman2015eclipse}.
This attack is aimed at a single node and consists in replacing all the peers of the victim with others controlled by the attacker.
Eclipsing can be used to perform other attacks, such as censorship, where the attacker drops transactions coming from the victim, or double spends.
By eclipsing a miner, an attacker could also perform selfish mining or even increase the chances of a 51\% attack~\cite{eyal2014majority}.

Despite being cited by virtually all topology-related papers as one of the attacks that justifies hiding the topology, the eclipse attack does not make use of any topology information besides the IP address of the victim.
In fact, this attack is mostly related to the address management and the mechanism used to establish new outbound connections.
Furthermore, a number of defensive mechanisms have been introduced in the Bitcoin reference client to avoid these attacks, such a the limitations in the number of outgoing connections to the same subnet or the randomization in the address selection procedure \footnote{See https://github.com/bitcoin/bitcoin/blob/v0.10.1/doc/release-notes.md .}.
Additionally, miners further protect themselves by running highly-connected gateway nodes~\cite{miller2015discovering} and propagating blocks via relay networks.

\subsection{Deanonymization}
The second major threat commonly linked to the knowledge of the network topology is transaction anonymity.
Several papers showed it is possible to deanonymize a specific transaction by detecting the node that generated it~\cite{biryukov2014deanonymisation,koshy2014analysis,fanti2017deanonimization}.

The basic idea, originally proposed by Kaminsky~\cite{kaminsky2011black}, consists in connecting to all reachable peers and monitoring incoming transactions. 
Intuitively, the first node to send a transaction will likely be the one that created it.
This attack only works against reachable peers, as the attacker needs to be connected to all possible sources.

By looking at known attacks, it can be observed that although topology information might improve the heuristics used to link transactions to their source node, it is hardly a requirement to deanonymize transactions.
The only exception to this is \cite{biryukov2014deanonymisation}, where unreachable nodes are deanonymized by identifying the set of their peers. 
However, knowing the topology of reachable nodes is irrelevant in this case.
Additionally, this technique has been made ineffective in the current protocol.
Apart from \cite{biryukov2014deanonymisation}, none of the state-of-the-art techniques use topology information, proving how these attacks do not depend on it.

In fact, as pointed out by Fanti et al.~\cite{fanti2017deanonimization}, the problem with anonymity lies in the transaction propagation algorithm, and should be addressed by adopting alternative relay protocols, such as Dandelion \cite{venkatakrishnan2017dandelion,fanti2018dandelion++} or \cite{franzoni2020improving}.

\section{Related work}
\label{sec:relatedwork}

To the best of our knowledge, there are no known algorithms to reliably compute the topology of a P2P network.
In fact, this information seems not to be as relevant to most such networks as it is for cryptocurrencies.
The closest-related works are on routing protocols~\cite{xu2003hieras,oliveira2005performance} and location-aware topology studies~\cite{rostami2007topology,liu2004location}.
However, none of these works explicitly address topology discovery.
Most notably, the only known topology-inferring techniques are those aimed at the Bitcoin network.
In this section, we briefly review all such techniques.

The first technique was proposed in \cite{biryukov2014deanonymisation} and allowed to determine the outbound connections of a node, based on the propagation of \texttt{addr} messages: as nodes advertise their own address when connecting to a new peer, and such messages are forwarded to other peers, it was possible to detect the target's peers by connecting to all nodes and analyzing received addresses.
Similarly, in \cite{miller2015discovering}, Miller et al. proposed a network-wide technique, called \textit{AddressProbe}, that allowed inferring connections among reachable nodes.
Their technique exploited the timestamp attached to each entry of \texttt{addr} messages: since the address of each outbound peer was updated every time a message was received, it was possible to determine which entries in the \texttt{addr} message corresponded to outbound peers by looking at their timestamp.
Both techniques were made ineffective by an update in the reference client \cite{nick2015guessing}. 

A more generic technique was proposed by Neudecker et al. \cite{neudecker2016timing}, which was based on the rumor centrality of gossip propagation.
Assuming the source of a specific message is known and that the adversary is connected to all nodes, it is possible to infer connection by observing the time at which nodes propagate the information.
Their technique was made ineffective by the switch in the spreading protocol from Trickle to Diffusion (\cite{fanti2017anonymity}).

In \cite{grundmann2018exploiting}, Grundmann et al. propose two different methods.
The first one exploits the fact that nodes accumulate transactions before announcing them to their peers.
In particular, an \texttt{inv} message contains all transactions received since the last \texttt{inv} message was sent. Based on this fact, the adversary creates marker transactions for all peers and observes \texttt{inv} messages to infer links.
This technique shows high precision (more than 90\%) but has very low recall (10\%), and is hardly practical in real life.
The second method targets a single node, and exploits the fact that clients do not relay double-spending transactions.
To infer the target's peers, the adversary sends a different double-spending transaction to each node, except the target.
Then she observes which transaction the target relays and deduces a link with the node to which that transaction was sent.
Although this technique has very high precision (97\%) and good recall (60\%), no countermeasures have been introduced to date.

The most recently disclosed technique for topology inferring is \textit{TxProbe} \cite{delgado2018txprobe}.
Similar to the previous example, this technique leverages marker messages, this time based on orphan transactions (i.e., transactions spending unknown inputs).
Despite its high precision and recall (more than 90\%), this method is rather invasive and can interfere with the ordinary transaction propagation.
Again, the technique was invalidated by a recent update in the reference client \footnote{See "Select orphan transaction uniformly for eviction" (https://github.com/bitcoin/bitcoin/pull/14626)}.

All mentioned techniques target reachable nodes, and require the adversary to connect to the whole network and observe data propagation either in a passive way or by actively introducing marker messages to infer communication links.
In Section \ref{sec:design}, we will leverage these concepts to design our topology-monitoring solution.

\section{The AToM Protocol}
\label{sec:design}
We propose a dynamic topology-monitoring algorithm for the Bitcoin P2P network called AToM (Active Topology Monitor).
In this section, we give an overview of our protocol, explaining its operating principles, and motivating our design choices.

\paragraph{\textbf{Notation}}
We denote the Bitcoin network as a directed graph $G{=}(V,E)$, 
where $V{=}\{N_1, N_2, ...\}$ and $E{=}\{(N,P){\text{ }{:}\text{ }}N{\to}P\}$.
A generic node is denoted by $N{=}addr_N$, where $addr_N$ is the IP address accepting incoming connections.
We denote the set of outbound peers of a node $N$ as $O_N$ and the set of inbound peers as $I_N$.
A generic peer of $N$ is indicated by $P{=}(addr_P, out)$, where $out$ is $true$ if $P{\in}ON$, and $false$ if $P{\in}I_N$.
We represent an outbound connection from $N$ to $P$ with the notation $N{\to}P$.

We use $M$ to denote a monitor, and $G_M{=}(V_M,E_M)$ to indicate its local topology snapshot.
Finally, we use $\Gamma$ to denote the set of valid monitors running in $G$.

\subsection{Protocol Overview}
The AToM protocol is run by a set of \textit{monitors} that connect to all reachable nodes.
The monitors continuously run a topology-inferring protocol to maintain an updated snapshot of the network.
An example of this scenario is depicted in Figure \ref{fig:pocnetwork}: a monitor connects to all reachable nodes, and run the AToM protocol to compute and maintain a continuously up-to-date state of the topology; links between unreachable nodes and reachable nodes are not included in the computed snapshot.

\begin{figure}[t]
	\centering
	\includegraphics[width=0.55\columnwidth]{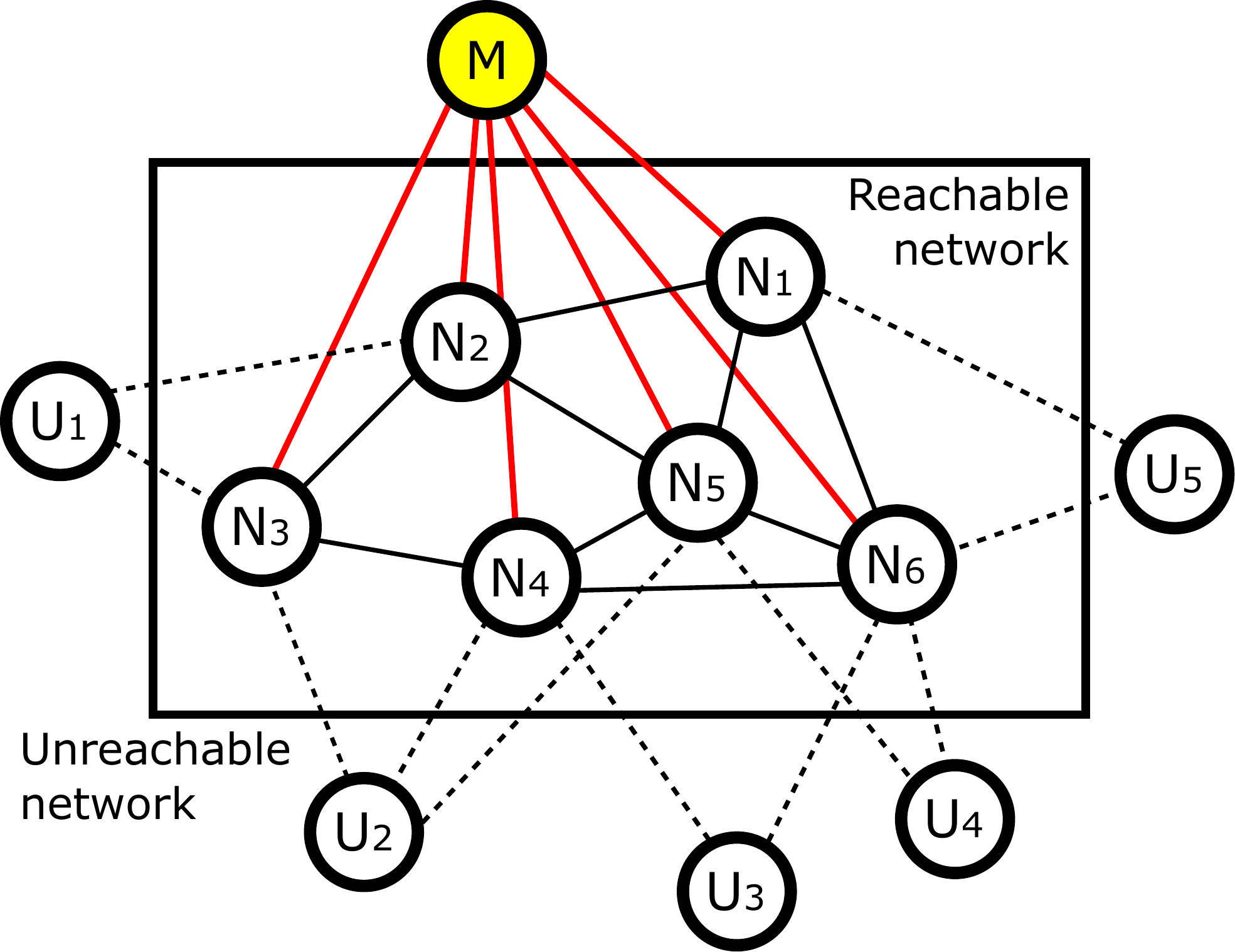}
	\caption[AToM: Scenario]{Our scenario: a monitor node ($M$) connects to all reachable nodes ($N$), excluding unreachable ones ($U$).}
	\label{fig:pocnetwork}
\end{figure}

We assume monitors know all public nodes. 
To that purpose, monitors can use a crawler or external services like Bitnodes \cite{bitnodes}.

\paragraph{\textbf{Scope}}
We limit the scope of the AToM protocol to the reachable part of the network.
There are different reasons behind this choice.

As mentioned in Section \ref{sec:background}, reachable nodes constitute the backbone of the network, since they maintain the vast majority of connections.
On the contrary, unreachable nodes are only marginally involved in data propagation \cite{wang2017towards}.
This arguably makes the reachable part the most important to protect and optimize.

From a security perspective,  
it is also safer to limit the public topology to reachable nodes, as this does not affect their protection from known attacks, as discussed in Section \ref{sec:topologysecurity}.
Instead, unreachable nodes might be more exposed to certain attacks if the adversary had access to this information.

Finally, there are several practical advantages in monitoring reachable nodes only.
First of all, as discussed in Section \ref{sec:background}, the reachable portion of the network is relatively small, and its nodes show more stability in terms of number and connections~\cite{bitnodes,statoshi-peers}.
This eases the monitoring task, which can better adapt to changes in the network.
Another advantage is the fact that is possible to guarantee the (almost) completeness of the snapshot computed by monitors.
In fact, it is virtually impossible to know (and connect to) all unreachable nodes, due to the inability to reach them from the Internet.
Because of this, malicious users could also fake the existence of unreachable nodes, without the monitors being able to verify them.
Instead, reachable nodes can be verified by simply opening a connection towards them and ensuring they run the Bitcoin protocol\footnote{Note that there is a one-to-one relation between reachable nodes and 'IP:port' tuples; in fact, while multiple nodes can run on the same device, only a single 'IP:port' address can be used for each running instance.}.

Note that, despite being out of scope, unreachable nodes could also make use of the AToM protocol by querying monitors a snapshot of the topology.
Specifically, each unreachable node can compute its relative snapshot by adding its own connections to the queried snapshot.
In fact, as these nodes only maintain outbound peers, all their connections are towards reachable nodes.
This extended snapshot could be used, for example, to autonomously decide which peers to connect to.

Given the above discussion, the restriction of AToM to reachable nodes should be considered a feature of the protocol, rather than a limitation.
In the rest of the paper, when talking about nodes, we will always refer to reachable ones, unless specified.

\subsection{Approach}
Similar to other topology-inferring techniques described in Section \ref{sec:relatedwork}, our solution has monitors connect to all nodes and leverage \textit{marker} messages to verify links.
To minimize the overhead, monitors verify outbound connections only.
While this might seem counter-intuitive, it is easy to see that this allows covering the whole (reachable) network.
In fact, 
all connections in the network are outbound, relatively to the node that initiated them, while inbound connections are just their symmetric view.
This approach also allows us to implicitly exclude unreachable nodes from the protocol, since no outbound connection can be established towards them.

To verify a link between two nodes, monitors have a marker message go through it, proving the two nodes are connected.
Adding an unpredictable random value to the marker, monitors ensure the only way a node can know it, is to have received it from the peer to which it was originally sent.
For instance, if a monitor wants to verify a connection between nodes $A$ and $B$, it sends to $A$ a marker containing a random value $r$; then, it probes $B$ for such a value.
If $B$ replies with the correct value, the monitor considers the connection verified.
In fact, the only way for $B$ to know $r$ is to have received from $A$, which proves they are connected.

Different from inferring techniques, which typically make use of side channels, we have nodes actively participate in the AToM protocol. 
This makes the result more reliable, but presents a downside:
if nodes misbehave, because faulty or malicious, it is hard to prove or disprove a connection without making use of trusted solutions like certificates or trusted execution environments.
We address this by assuming 
nodes have a list of "semi-trusted" monitors, which are well-known (and thus partially trusted) nodes in the network, and that the majority of them is honest.
This means the set of valid monitors should be agreed on beforehand by the Bitcoin community.
At a practical level, this can be obtained by leveraging already-existing semi-trusted entities of the Bitcoin network, such as the DNS servers used for node bootstrapping, or the list of peers hardcoded in the reference client~\cite{bitcoin-core}.

\subsection{The AToM Design}
\label{subsec:design}
Each monitor computes the network topology by executing, for each node, a \textit{Peers Verification} (PeeV) protocol.

The PeeV protocol verifies the outbound connections of a node using a \textit{marker} message.
This message contains the following information: the monitor $M$ that created the marker, the target node $N$ whose connections are being verified, and a random value $r$.
The monitor ID allows nodes to recognize which monitor is running the protocol and to verify it is a valid actor.
The target node ID and the random value are required to avoid malicious behaviors, as we will show later in this section.

Specifically, the PeeV protocol for a target node $N$ and a monitor $M$ works this way: 
\vspace{-10pt}
\begin{algorithm}[h]
	\caption{PeeV protocol for a monitor $M$ and a node $N$}
	\label{alg:peev}
	\begin{algorithmic}[1]
		\State $M$ creates a \textit{marker} message $\mu{=}[M,N,r]$, where $r$ is a random value;
		\State $M$ sends $\mu$ to $N$;
		\State $N$ forwards $\mu$ to its outbound peers $P_1,P_2,...$; 
		\State Each peer $P_i$ forwards $\mu$ to $M$;
		\State Upon receiving a marker $\mu'$ from $P_i$, $M$ checks whether $\mu'=\mu$:
		\State \hspace{10pt} If so, $M$ adds the connection $N{\to}P_i$ to its local topology snapshot.
	\end{algorithmic}
\end{algorithm}
\vspace{-10pt}
\begin{figure}[t]
	\centering
	\vspace{9pt}
	\includegraphics[width=0.3\textwidth]{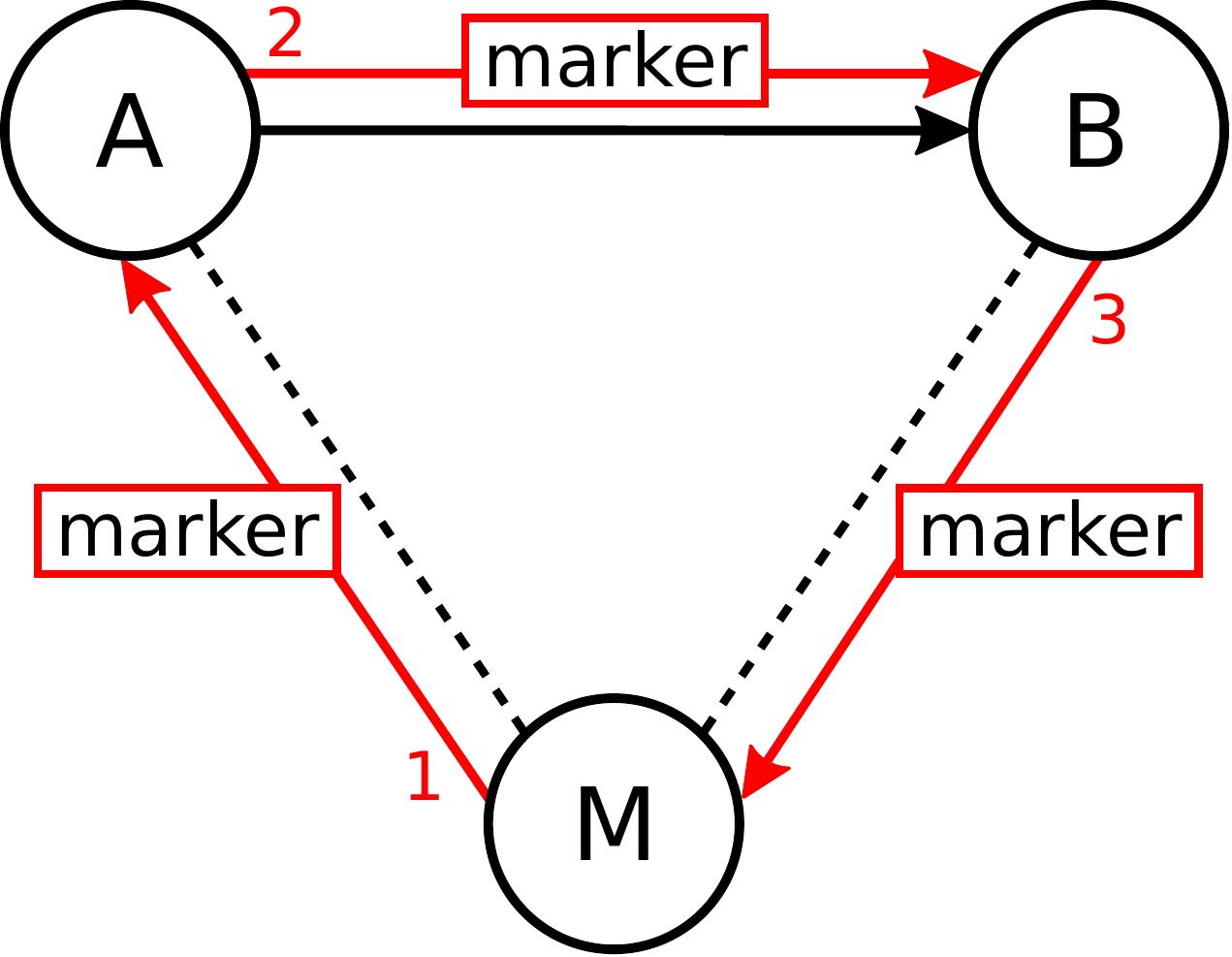}
	\vspace{5pt}
	\caption[AToM: Peers Verification overview]{PeeV overview: a monitor $M$ verifies $A{\to}B$ by sending a $marker$ message to $A$ and receiving it from $B$. Red arrows show the route of the marker.}
	\label{fig:poc}
\end{figure}

We call a single execution of this protocol a \textit{PeeV round}.
An example of a PeeV round verifying a connection $A{\to}B$ is depicted in Fig. \ref{fig:poc}.

\paragraph{\textbf{Handling network changes}}
By running a PeeV round for every node, a monitor obtains a snapshot of the full network topology.
However, changes in the network can occur at any time, producing errors in the computed snapshot.
While the relative stability of the reachable network makes the number of mistakes in a single snapshot very limited, information on the connections among nodes should be updated over time in order to monitor the network.

An easy way for monitors to do so is to simply scan the network (i.e., to run the PeeV protocol for all nodes) at a certain frequency, keeping the last computed results as the up-to-date snapshot.
However, this approach cannot adapt well to all nodes.
In fact, each node can experience changes at a different rate, making it hard to decide an appropriate scan frequency for the whole network.
In particular, if the frequency is too high, it might affect efficiency and increase network traffic.
On the other hand, setting this value too low would affect the accuracy of the snapshot.
For instance, short-lived connections established within two consecutive scans would remain undetected.

Therefore, we adopt a continuous-monitoring approach,
scanning each node (with the PeeV protocol) individually, with an independent frequency value.
In turn, this value is dynamically adjusted for each node according to the stability of its connections (i.e., the number of changes per unit of time).
This way, nodes experiencing changes at a higher rate will be scanned more frequently than stable ones.
Such an approach allows us to improve accuracy, as the scanning process adapts to the variability of single nodes, and reduce the impact on the network, as message exchanges are spread over time (while performing a one-shot snapshot requires exchanging all messages at once).

\subsubsection{Securing the protocol}
\label{subsec:malicious}
A potential issue for the AToM protocol is the presence of \textit{misbehaving nodes}.
These nodes can deviate from protocol, either accidentally (if faulty) or intentionally (if malicious), producing errors in the snapshot.
In particular, while faulty nodes do not behave consistently with each other (making it easier to spot the error), malicious ones, when controlled by a single actor, can cooperate to deceive monitors.
Therefore, in this section, we specifically address malicious behaviors, with the goal of preventing them from affecting the accuracy of the protocol.
While our countermeasures are designed to protect from malicious nodes, it should be clear that they are effective against faulty nodes as well.

\paragraph{\textbf{Adversary model}}
We consider an adversary that controls an arbitrary number of colluding nodes and aims at affecting the accuracy of the snapshot computed by the monitors.
To that purpose the adversary can try to hide or fake nodes, and to hide or fake connections.

To hide a node, the adversary needs to avoid all connections from the monitors (whose ID is known).
This can be done, for instance, by blacklisting monitor addresses.
However, monitors can easily bypass this restriction by connecting from different (unknown) addresses.
Thus, the only way a node can hide completely is by rejecting all inbound connections.
Nevertheless, this would effectively make the node unreachable, thus falling out of the AToM scope.

To fake a node, the adversary can announce a fake address, pretending the existence of a running node.
However, as previously mentioned, monitors verify the existence of a node by connecting to them and running the Bitcoin protocol.
Therefore, the only way to fake a node is to have a single instance of a client accepting connections through multiple addresses (this can be done by using different ports or different IPs redirected to the running node).
This case is analogous to controlling multiple colluding nodes (communicating with each other), which is part of our adversary model.

Colluding nodes can be used by the adversary for hiding or faking connections.
In particular, two such nodes can easily hide a connection among each other, or fake one by using external channels or a third colluding node.
There are virtually no means to detect or avoid this kind of behavior, as it is not possible to prevent malicious nodes from cooperating.
However, there is no clear advantage for the adversary in doing so, since she controls both edges of the connection.
Instead, it is realistic to assume the adversary would try to hide or fake connections with honest nodes, as this might be used for other attacks (e.g., eclipse).
Therefore, when considering a pair of nodes, we will only focus on the case where at least one node is honest.

\paragraph{\textbf{Malicious behaviors}}
Malicious nodes can deviate from the protocol by doing any of the following actions:
\begin{enumerate}
	\label{atom:malicious}
	\item forward a marker to an inbound peer; 
	\item forward a marker to a non-peer node, through a colluding node; 
	\item reuse a marker from a previous PeeV round (replay attack); 
	\item tamper with a marker; 
	\item drop a marker received from the monitor (i.e., not forward the message to one or more outbound peers); 
	\item drop a marker received from a peer (i.e. not forward the message to the monitor).
\end{enumerate}

It is easy to see that each of the above-listed behaviors potentially produce errors (false positives and false negatives) in the snapshot computed by a monitor.
In particular, cases (1) and (2) might induce the monitor to infer a non-existing outbound connection (false positive).
Similarly, case (3) might make the monitor keep in the topology snapshot a connection that does not exist anymore (false positive).
The effect of case (4) depends on which field the adversary modifies.
On the one hand, changing the monitor identity or the random value would make the marker being dropped or rejected (false negative); as such, this analogous to dropping the message (cases 5 and 6), although it generates more traffic.
On the other hand, changing the target field could be used in combination with case (2) to produce a fake connection (false positive). 
However, since the random value is bound to the target node, this modification would result in the marker being considered invalid. 
Cases (5) and (6) effectively hide one or more connections from the monitor (false negative).

Given the importance of accuracy in our topology-computation protocol, it is necessary to introduce countermeasures to avoid the effects of the above-listed behaviors.

\paragraph{\textbf{Handling malicious behaviors}}
To ensure the accuracy of our protocol in the presence of malicious nodes, we address each of the behaviors listed above.

To avoid case (1), we make (honest) nodes accept markers only from inbound peers; markers received from outbound peers are simply dropped.
Similarly, to handle case (2), we make nodes accept markers only when received from the specified target.
To avoid replay attacks (case 3), we make monitors accept only markers of the current PeeV round (for a specific target).
In this respect, the random value $r$ acts as an identifier of the round.

Case (4), (5) and (6) are hard to avoid, since we cannot prevent a malicious node from dropping or modifying a message \footnote{Note that adding a digital signature to the marker (to mitigate case 4) is not helpful, as this would be dropped or rejected either way. However, it could prevent honest nodes from forwarding invalid markers (in the combination of cases 4 and 2).}.
Therefore, we mitigate these cases by leveraging information from the monitors.

In particular, we make each node maintain a peer if confirmed by the majority of monitors.
At the end of a PeeV round, monitors send to the target node the list of its currently verified peers (both inbound and outbound).
Whenever a peer is not confirmed by the majority of monitors, its connection is closed, and the monitor is blacklisted (typically, in Bitcoin, nodes are banned for 24 hours before being readmitted as peers).

This mechanism aims at discouraging malicious nodes from misbehaving, as it might make them lose connections.
Additionally, it allows monitors to restore the correctness of the topology snapshot.
In fact, while the snapshot is initially incorrect for not including an existing connection (false negative), it becomes correct as soon as such connection is closed (true negative).
We refer to this feature as \textit{enforced consistency}.
Note that such a behavior only applies to connections where the malicious node is connected to an honest peer, since, as already discussed, it is not possible to prevent two malicious nodes from faking or hiding a connection.

\paragraph{\textbf{Handling Man-In-The-Middle Attacks}}
As connections in Bitcoin are not encrypted, a MitM adversary is able to drop, modify, and forge messages.
As we saw, the use of random values in markers prevents the adversary from forging or modifying valid markers.
In fact, these attacks are akin to the cases described above, and ultimately result in a connection not being verified.
However, differently from previous cases, this attack might aim at hiding a connection between honest nodes, eventually leading to the connection to be lost.
In other words, this is can be seen as a DoS attack.

Similarly to the drop case, there is no possible countermeasure to avoid the attack.
Nonetheless, it is worth noting that a MitM attacker is always able to drop the connection it controls, and, in the case of the Bitcoin protocol, other dangerous attacks are possible, such as deanonymization and double-spending.
Therefore, in that respect, our protocol does not create any additional attack surface.

\section{Protocol Procedures}
\label{protocol}
In this section, we formally define the AToM protocol by means of the procedures executed by monitors and nodes.
For the sake of simplicity we refer to a single monitor when describing the procedures, unless where differently specified.

\paragraph{\textbf{Protocol Messages}}
We introduce the following protocol messages:
\begin{itemize}
	\setlength\itemsep{0em}
	\item \texttt{marker = [target,monitor,value]}: sent by a monitor $M$ to a node $N$, where $\texttt{target}{=}N$, $\texttt{monitor}{=}M$ and $\texttt{value}$ is a random number;
	\item \texttt{verified = [$L_N^M$]}: sent by a monitor $M$ to a node $N$, where $L_N^M$ is the list of peers of $N$ currently verified by $M$.
\end{itemize}

\paragraph{\textbf{PeeV procedures}}
To run the PeeV protocol for a target node $N$, monitors execute the procedure shown in Algorithm \ref{alg:poc}.
This procedure generates a random number $r$, and sends to $N$ a \verb|marker| message with the following payload: $$\texttt{marker := [target:N,monitor:M,value:r]}.$$
To avoid indefinite waits, the monitor sets a timeout \verb!t! for receiving markers back from other nodes.
When a \verb!marker! message is received from a node $P$, it is checked against the one sent to $N$.
If it matches, $P$ is added to the list of verified peers $L_N^M$.
When \verb!t! expires, the PeeV round ends, and the procedure outputs the list $L_N^M$.

When a \verb|marker|${=}[M,N,r]$ message is received, nodes process it with the $HandleMarker$ procedure, shown in Algorithm \ref{alg:handle-marker}.
This procedure acts depending on the source of the message ($pfrom$):
if $pfrom$ is a legit monitor $M$, the \verb|marker| is forwarded to all outbound peers;
if $pfrom$ is an inbound peer and corresponds to the target $N$, the \verb|marker| is forwarded to the monitor $M$.

\begin{figure}[tp]
	\vspace{-20pt}
	\begin{minipage}{.48\textwidth}
		\begin{algorithm}[H]
			\caption{PeeV}
			\label{alg:poc}
			\begin{algorithmic}[1]
				\Statex $PeeV(N)$:
				\IState $r \gets rand()$ 
				\IState $\texttt{marker} {=} [N,M,r]$
				\IState $send(N,\texttt{marker})$
				\IState $t \gets getTimeout(N)$
				\IState \textbf{while} $now {<} t$ \textbf{do}:
				\IIState $receive(P,\texttt{marker})$ 
				\IIState \textbf{if} $\texttt{marker} == [N,M,r]$:
				\IIIState  $L_N^M := L_N^M \cup \{P\}$
				\IState \textbf{output} $L_N^M$
			\end{algorithmic}
		\end{algorithm}
	\end{minipage}
	\hspace{10pt}
	\begin{minipage}{.48\textwidth}
		\begin{algorithm}[H]
			\caption{HandleMarker}
			\label{alg:handle-marker}
			\begin{algorithmic}[1]
				\Statex \textbf{Env:} $O$, $\Gamma$		
				\Statex $HandleMarker(pfrom,\texttt{marker}):$
				\IState $[N,M,r] {=} \texttt{marker}$
				\IState \textbf{if} $pfrom == M$ \texttt{and} $M{\in}\Gamma$:
				\IIState \textbf{for each} $P$ in $O$ \textbf{do}:
				\IIIState $send(P,\texttt{marker})$
				\IState \textbf{if} $pfrom == N$ \texttt{and} $pfrom\text{ in }I$:
				\IIIState $send(M,\texttt{marker})$
			\end{algorithmic}
		\end{algorithm}
	\end{minipage}
\end{figure}

\paragraph{\textbf{AToM procedures}}
To build and maintain an up-to-date snapshot of the topology ($G_M$), monitors run the $AToM$ procedure, shown in Algorithm \ref{alg:atom}.
This procedure executes a loop for every node $N$:
at each iteration, the $PeeV$ procedure is run, and the output $L_N^M$ is used to update the snapshot $G_M$ ($updateTopology()$).
A \verb!verified! message is then sent to $N$ with the list of the verified peers, that is, $L_N^M$ plus all known inbound peers currently verified by other PeeV executions ($getPeers()$).
Note that peer lists do not change between PeeV rounds, which means that a connection stays in $G_M$ until a PeeV execution fails to verify it.

Before the following PeeV round for $N$, a time $f_N$ is waited ($waitNext()$), which depends on the corresponding scan frequency.
The value of $f_N$ is specific to the node $N$, and is adjusted ($adjustFrequency()$) at the end of each PeeV round according to the number of changes in the verified peer list $L_N^M$. 
In particular, when no change is detected, $f_N$ is increased (that is, the scan frequency is reduced);
instead when $c>1$ changes are detected, $f_N$ is reduced by an amount proportional to $c$.
When only 1 change is detected, the frequency is kept unchanged.
This mechanism allows AToM to better adapt to sudden and isolated pikes in the variability of peer connections.

If $N$ disconnects, it is removed from $V_M$, and all its connections are removed from $E_M$ ($removeNode()$).
Again, we assume new nodes are automatically discovered by the monitor, and added to $V_M$.
When this occurs, the corresponding PeeV loop is started.

\begin{figure}[t]
	\vspace{-20pt}
	\begin{algorithm}[H]
		\small
		\caption{AToM}
		\label{alg:atom}
		\begin{multicols}{2}
			\begin{algorithmic}[1]
				\Statex \textbf{Env}: $G_M{=}(V_M,E_M)$,  $f_{min}$, $f_{max}$
				\Statex
				
				\setcounter{ALG@line}{0}				
				\Statex $AToM()$:
				\State \textbf{for each} $N$ in $V_M$: 
				\IState \textbf{while} $isOnline(N)$ \textbf{do}:
				\IIState $L_N^M \leftarrow PeeV(N)$
				\IIState $updateTopology(E_M, L_N^M)$
				\IIState $\texttt{verified} = getPeers(N,E_M)$
				\IIState $send(N, \texttt{verified})$
				\IIState $adjustFrequency(N,L_N^M)$
				\IIState $waitNext(f_N)$
				\IState \textbf{done}
				\IState $remove(N)$		
				
				\vfill\null
				\columnbreak
				
				\setcounter{ALG@line}{0}				
				\Statex $updateTopology()$:
				\IState \textbf{for each} $P$ \textbf{in} $V_M$:
				\IIState \textbf{if} $P$ \textbf{in} $L_N^M$:
				\IIIState   $E_M = E_M \cup \{(N,P)\}$
				\IIState \textbf{else}:
				\IIIState 	 $E_M = E_M - \{(N,P)\}$		
				
				\Statex
				\setcounter{ALG@line}{0}				
				\Statex $adjustFrequency(N,L_N^M)$:
				\IState $\texttt{c} \leftarrow countChanges(G_M,L_N^M)$
				\IState \textbf{if} $\texttt{c} == 0$ \textbf{and} $f_N < f_{max}$: 
				\IIState $f_N = f_N + 1$
				\IState \textbf{else if} $\texttt{c} {>} 1$ \textbf{and} $f_N >= f_{min}$:
				\IIState $f_N = f_N - \texttt{c}$
				\Statex
				
			\end{algorithmic}
		\end{multicols}
		\vspace{-15pt}
	\end{algorithm}
	\vspace{-30pt}
\end{figure}

\begin{figure}[t]
	\begin{algorithm}[H]
		\small
		\caption{HandleVerified}
		\label{alg:handle-verified}
		\begin{multicols}{2}
			\begin{algorithmic}[1]
				\Statex \textbf{Env:} $Peers{=}O \cup I$, $\Gamma$, $\tau{=}\frac{|\Gamma|}{2}$
				\Statex 
				
				\setcounter{ALG@line}{0}	
				\Statex $HandleVerified(M, \texttt{verified}):$
				\IState $L_N^M = \texttt{verified}$
				\IState \textbf{for each} $P$ in $Peers$:
				\IIState \textbf{if} $P$ in $L_N^M$: $\upsilon_P^M = 1$
				\IIState \textbf{else}: $\upsilon_P^M = 0$	
				\IIState $checkReputation(P)$
				
				\vfill\null
				\columnbreak
				\Statex \Statex
				
				\setcounter{ALG@line}{0}	
				\Statex $checkReputation(P)$:
				\State $\phi_P = 0$
				\State \textbf{for each} $M$ in $\Gamma$: 
				\IState $\phi_P = \phi_P + \upsilon_P^M$
				\State \textbf{if} ($\phi_P \leq \tau$)
				\IState $disconnect(P)$		
			\end{algorithmic}
		\end{multicols}
		\vspace{-15pt}
	\end{algorithm}
	\vspace{-10pt}
\end{figure}

\paragraph{\textbf{The peer reputation system}}
Upon receiving a \verb!verified! message from a monitor $M$, nodes execute the $HandleVerified$ procedure shown in Algorithm~\ref{alg:handle-verified}.
This procedure uses the list of verified peers $L_N^M$ included in the message to maintain a reputation system.

In particular,
each peer has a reputation value $\phi_P$ based on the \textit{verification status} $\upsilon_P^M$ of each monitor $M$.
To set this value, the $HandleVerified$ procedure 
checks the list of current peers ($Peers$) against the list of verified peers received from $M$ ($L_N^M$).
If a peer $P$ is in $L_N^M$, $\upsilon_P^M$ is set to $1$, otherwise, it is set to $0$.
The reputation for a peer $P$ ($\phi_P$) is then obtained by summing the confirmation status of all monitor.

When connecting to a peer $P$, its verification status is set to $1$ for all monitors, so that its initial reputation value is $\phi_P{=}|\Gamma|$, corresponding to the number of monitors.
This value is then updated every time a \verb|verified| message is received.
If the reputation value falls below a threshold $\tau{=}\frac{|\Gamma|}{2}$, the peer is disconnected.
Therefore, a connection is maintained as long as it is confirmed by the majority of monitors.

This reputation system is designed to prevent the attacker from keeping a connection hidden and alive at the same time, effectively making the attack inconvenient.

\paragraph{Global Topology}
We consider as a full trusted snapshot of the topology the union of nodes and connections confirmed by the majority of the monitors.

We first define the \textit{verification set} for a connection $A{\to}B$ as the set of monitors that verified such a connection. Formally:
\begin{definition}
	\label{def:verification}
	Let $G{=}(V{,}E)$ be a Bitcoin network, $\Gamma$ be the set of monitors in $G$, and $(A{,}B)$ be a connection in $E$, 
	the \textit{verification set} $\Sigma_{AB}$ of $(A,B)$ is
	$$\Sigma_{AB}=\{M : M{\in}\Gamma \wedge (A,B){\in}E_M\}.$$
\end{definition}
Then, we define the \textit{global snapshot} $G_{AToM}$ as the union of the connections that are verified by the majority of monitors.
Formally:
\begin{definition}
	\label{def:gatom}
	Let $G{=}(V{,}E)$ be a Bitcoin network, and $\Gamma$ be the set of monitors in $G$,
	the \textit{global snapshot} $G_{AToM}$ computed over the network $G$ is
	$$G_{AToM} = (V, E_{AToM}),$$
	where
	$$E_{AToM} = \{(A,B) : |\Sigma_{AB}| > {|\Gamma|/2}\}.$$
\end{definition}
We assume monitors synchronize among each other to compute $G_{AToM}$ and always use this snapshot as the trusted one.
In particular, whenever some party requests the current state of the topology to a monitor, the global snapshot $G_{AToM}$ is returned. 

\section{Analysis}
\label{sec:analysis}

In this section, we study the correctness and accuracy of our protocol, both in an honest setting and in the presence of misbehaving nodes.
Finally, we analyze the overhead for participating nodes in terms of the number of extra messages exchanged.

For our analysis, we consider a monitor $M$ and a network $G{=}(V,E)$.
Without loss of generality, we assume $M{\in}V$ and $(M,N){\in}E$ for every $N$ in $V$, but do not show them in $G$.
As $M$ is not included in the snapshot $G_M$, this does not affect our analysis.
We analyze PeeV and AToM by means of their procedures, described in Algorithm \ref{alg:poc} and Algorithm \ref{alg:atom}, respectively.

\subsection{Correctness}
\label{sec:correctness}
For the sake of simplicity, we show the correctness of AToM in a trusted setting (i.e., when all nodes are honest).
In Section \ref{sec:security}, we will study how misbehaving nodes can affect the correctness of the result.
Additionally, we assume a connection is never dropped during the execution of the protocol (we will relax this assumption in the accuracy analysis).

Showing the correctness of our protocol in a trusted environment is relatively straightforward.
In particular, we can easily prove that a connection $N{\to}P$ is added to a topology snapshot $E_M$ if and only if it exists in the actual topology (i.e., $(N,P){\in}E$).

\paragraph{\textbf{PeeV}}
We prove that, for a target node $N$ and a peer $P$, the $PeeV$ procedure adds $P$ to $L_N^M$ if and only if $N{\to}P$ exists in $G$.
We give a formal proof for a network $G{=}(V,E)$, and nodes $N,P{\in}V$.

\begin{theorem}
	\label{theorem:peev}
	Let $G{=}(V,E)$ be a Bitcoin network, $N,P$ be nodes of $V$, and $M$ a legit monitor.
	Then, if $M$ executes the procedure $PeeV(N)$ defined in Algorithm \ref{alg:poc}, and $N,P$ execute the $HandleMarker()$ procedure defined in Algorithm \ref{alg:handle-marker}, the following condition holds on the output $L_N^M$ of $PeeV(N)$:
	$$P{\in}L_N^M \iff (N,P){\in}E.$$
\end{theorem}
\begin{proof}
	We prove the two sides of the implication separately.	
	We assume the connection $N{\to}P$ is maintained during the whole execution of the protocol, and that the timeout $t$ in $PeeV(N)$ is longer than the time needed by a \texttt{marker} message to be transmitted through the sequence of nodes $M{-}N{-}P{-}M$. 
	
	We start proving that $(N,P){\in}E {\implies} P{\in}L_N^M.$
	
	Let $L$ be the set of outbound peers of $N$, and 
	let us assume $(N,P){\in}E$.
	Then, if $M$ executes $PeeV(N)$, and $N$,$P$ execute $HandleMarker()$, the following occurs:
	$M$ creates \verb|marker|${=}[N,M,r]$ and sends it to $N$;
	$N$ receives \verb|marker| from $M$, and since $\texttt{pfrom}{=}M$, it sends \verb|marker| to every peer in $L$;
	since $(N,P){\in}E {\implies} P{\in}L$, $P$ receives \verb|marker| from $N$;
	as $\texttt{pfrom}{=}N$, $P$ sends \verb|marker| to $M$;
	$M$ receives \verb|marker| from $P$;
	as $\texttt{marker}{=}[N,M,r]$, $M$ executes $L_N^M{:=}L_N^M{\cup}{P}$;
	finally, the $PeeV(N)$ procedure outputs $L_N^M$, such that $P{\in}L_N^M$.
	
	We prove the converse side of the implication by proving its contrapositive: $(N,P){\notin}E {\implies} P{\notin}L_N^M.$
	Let $L$ be the set of outbound peers of $N$, and let us assume $(N,P){\notin}E$.
	Then, if $M$ executes $PeeV(N)$, and $N$ executes $HandleMarker()$, the following occurs:
	$M$ creates \verb|marker|${=}[N,M,r]$ and sends it to $N$;
	$N$ receives \verb|marker| from $M$, and since $\texttt{pfrom}{=}M$, it sends \verb|marker| to every peer in $L$;
	since $(N,P){\notin}E {\implies} P{\notin}L$, $P$ does not receive \verb|marker| from $N$;
	as $M$ does not receive \verb|marker| from $P$, it does not execute $L_N^M{:=}L_N^M{\cup}{P}$;
	finally, the $PeeV(N)$ procedure outputs $L_N^M$, such that $P{\notin}L_N^M$.

\end{proof}

\paragraph{\textbf{AToM}}
The correctness of AToM directly derives from that of PeeV.
In particular, we can prove that the $AToM$ procedure adds a connection $N{\to}P$ to the snapshot $G_M$ if and only if $(N,P){\in}E$.
Similar to the previous case, we assume, for the sake of simplicity, that connections are not dropped.

\begin{theorem}
	\label{thr:atom-correctness}
	Let $G{=}(V,E)$ be a Bitcoin network, and $M$ a legit monitor, and $G_M{=}(V,E_M)$ be the local snapshot of $M$.
	Then, if $M$ executes the procedure $AToM(V)$ defined in Algorithm \ref{alg:atom}, the following condition holds:
	
	$$(N,P){\in}E \iff (N,P){\in}E_M.$$
\end{theorem}
\begin{proof}
	Let us assume $(N,P){\in}E$.
	Then, when $M$ executes $AToM(V)$, the following occurs:
	since $N{\in}V$, $M$ executes a \texttt{while} loop for $N$;
	the $AToM$ procedure loop executes $L_N^M \leftarrow PeeV(N)$;
	the $AToM$ procedure loop executes $updateTopology(E_M,L_N^M)$;
	since, by Theorem \ref{theorem:peev}, $(N,P){\in}E \iff P{\in}L_N^M$, $updateTopology(E_M, L_N^M)$ adds $(N,P)$ to $E_M$ if and only if $(N,P){\in}E$.
	
\end{proof}

\subsection{Security Analysis}
\label{sec:security}
In this section, we analyze the correctness of AToM in the presence of misbehaving nodes. 
In particular, we study how malicious nodes can affect the global snapshot $G_{AToM}$.
We refer to the possible misbehaviors for a malicious node, listed in Section \ref{atom:malicious}: 
(1) forward a marker to an inbound peer; 
(2) forward a marker to a non-peer node, through a colluding node;
(3) resend a valid marker used in a previous PeeV round (replay attack); 
(4) tamper with a marker; 
(5) drop a marker received from the monitor (i.e., not forward the message to one or more outbound peers);
(6) drop a marker received from a peer (i.e. not forward the message to the monitor).

Given the fact that two colluding nodes cannot be prevented from faking or hiding connections, we only study the cases where at least one node is honest.
In Section \ref{sec:experiments}, we experimentally evaluate how accuracy is affected by how multiple colluding nodes can affect the accuracy of the protocol.

We consider a malicious node $N$ and an honest peer $P$.
For cases (1) to (3), we show that, for any monitor $M$ computing a snapshot $G_M$, $(N,P){\in}E_M$ only if $(N,P){\in}E$. We do so by proving that none of these behaviors make the PeeV protocol verify an incorrect peer.

\begin{lemm}
	Let $M$ be a monitor, $N$ be a malicious node, and $P$ be an honest inbound peer of $N$.
	If $M$ executes $PeeV(N)$ and $N$ forwards \texttt{marker} to $P$, then the following condition holds when $PeeV(N)$ ends: $(P){\notin}L_N^M$.
\end{lemm}
\begin{proof}
	Let us assume $M$ executes $PeeV(N)$ (Algorithm \ref{alg:peev}), and $N$ forwards to $P$ the \verb!marker! received from $M$.
	Then, the following occurs: $P$ receives \texttt{marker} and runs the $HandleMarker$ procedure (Algorithm \ref{alg:handle-marker}); 
	in this procedure, since $pfrom = N$, the \texttt{if} statement at line 5 yields $true$;
	since $pfrom.out = true$, the \texttt{if} statement at line 6 yields $false$, so line 7 ($send(M,\texttt{marker})$) is not executed;
	during the execution of $PeeV(N)$, since \texttt{marker} is not received from $P$, line 8 ($L_N^M := L_N^M \cup \{P\}$) is not executed;
	hence, at the end of $PeeV(N)$, $P{\notin}L_N^M$.
	
\end{proof}

\begin{lemm}
	Let $M$ be a monitor, $N$ be a malicious node, $C$ be a colluding node, and $P$ be an honest outbound peer of $C$ that is not connected to $N$.
	If $M$ executes $PeeV(N)$, $N$ forwards \texttt{marker} to $C$, and $C$ forwards it to $P$, then the following condition holds when $PeeV(N)$ ends: $P{\notin}L_N^M$.
\end{lemm}
\begin{proof}
	Let us assume $M$ executes $PeeV(N)$ (Algorithm \ref{alg:peev}), $N$ forwards \verb!marker! to $C$, and $C$ forwards it to $P$.
	Then, the following occurs: $P$ receives \texttt{marker} and runs the $HandleMarker$ procedure (Algorithm \ref{alg:handle-marker}); 
	in this procedure, since $pfrom = C$, the \texttt{if} statement at line 5 yields $false$, so line 7 ($send(M,\texttt{marker})$) is not executed;
	during the execution of $PeeV(N)$, since \texttt{marker} is not received from $P$, line 8 ($L_N^M := L_N^M \cup \{P\}$) is not executed;
	hence, at the end of $PeeV(N)$, $P{\notin}L_N^M$.
	
\end{proof}
Note that, as previously mentioned, the adversary could modify the \texttt{marker} message setting the colluding node $C$ as the $target$. In this case, $P$ would consider the message valid and forward it to $M$.
However, the $value$ field would not match the target, making the monitor discard the marker. 

\begin{lemm}
	Let $M$ be a monitor, $N$ be a malicious node, and $P$ be a previously-connected honest inbound peer of $N$.
	Let $PeeV(P)_i$ be a past PeeV round for node $P$, and let $\mu=[P,M,r_i]$ be the $\texttt{marker}$ received by $N$ in that round.
	Then, if during a subsequent round $PeeV(P)_j$, $N$ is not connected to $P$ (i.e., $(P,N){\notin}E$) and sends $\mu$ to $M$, the following condition holds when $PeeV(P)_j$ ends: $N{\notin}L_P$.
\end{lemm}
\begin{proof}
	Let us suppose $M$ is running $PeeV(P)_j$ (Algorithm \ref{alg:peev}) with a \texttt{marker} $\mu'=[P,M,r_j]$, and $N$ sends $\mu$ to $M$.
	Then, the following occurs: $M$ receives $\mu$ from $N$, a compare it with $\mu'$;
	since $\mu{\neq}\mu'$, line 8 ($L_N^M := L_N^M \cup \{P\}$) is not executed;
	hence, at the end of $PeeV(P)_j$, $N{\notin}L_P$.
	
\end{proof}

For cases (4) to (6), we show that a mismatch between $E$ and $E_{AToM}$ can only occur during a limited time frame, thanks to the enforced consistency feature.
For the sake of simplicity, we only prove case (5), as the other two cases are equivalent (modified markers are dropped by the monitor) and thus follow the same reasoning.
We assume the connection $N{\to}P$ can only be closed by the $HandleVerifier$ procedure.
\begin{lemm}
	Let $\Gamma$ be the set of monitors, $N$ be an honest node, and $P$ be a malicious outbound peer of $N$ ($(N,P){\in}E$).
	If each monitor $M$ runs $PeeV(N)$ (Algorithm \ref{alg:peev}) 
	then when all PeeV rounds have ended, the following condition holds:
	
	$$(N,P){\in}E_{AToM} \iff (N,P){\in}E$$
\end{lemm}
\begin{proof}
	Let us suppose $(N,P){\in}E$, and all monitors in $\Gamma$ execute $PeeV(N)$.
	For each monitor $M$, $N$ receives a \verb|marker| message during the PeeV round, which is forwarded to $P$ ($|\Gamma|$ messages in total).
	At the end of each PeeV round, $N$ receives a \verb|verified| message from each monitor $M$, containing the list $L_N^M$ of verified peers.
	
	Now, let $\psi$ be the number of markers that $P$ correctly forwards to the monitor.
	Given the threshold value $\tau{=}\frac{|\Gamma|}{2}$, 
	two cases are possible: (a) $\psi {>} \tau$ or (b) $\psi {\leq} \tau$.
	
	In case (a), 
	given Definition \ref{def:gatom}, we have $(N,P){\in}E_{AToM}$.
	At the same time, since $P{\in}L_N^M$ for $\psi{>}\tau$ monitors, we have $\phi_P{>}\tau$; 
	hence, $P$ is not disconnected ($(N,P){\in}E$).
	
	In case (b), 
	given Definition \ref{def:gatom}, we have $(N,P){\notin}E_{AToM}$.
	At the same time, since $P{\in}L_N^M$ for $\psi{\leq}\tau$ monitors, we have $\phi_P{\leq}\tau$; 
	hence, $P$ is disconnected ($(N,P){\notin}E$).

\end{proof}

\subsection{Accuracy}
As previously discussed, the accuracy of AToM depends on how the network topology changes over time.
In this section, we study all the events that can produce a mismatch between the global snapshot $G_{AToM}$ and the actual topology $G$.
In particular, we consider the following events: 
\begin{enumerate}
	\item a new node $N$ joins the network: $V := V \cup \{N\}$;
	\item a node $N$ leaves the network: $V := V - \{N\}$;
	\item a new connection $N{\to}P$ is established: $E := E \cup \{(N,P)\}$;
	\item a connection $N{\to}P$ is closed: $E := E - \{(N,P)\}$.
\end{enumerate}
Note that case (1) implies case (3), since, by definition, a node with no connections is not part of the network. 
Similarly, case (2) implies case (4) for all the connections held by the node leaving the network.
Thus, without loss of generality, we focus on cases (3) and (4).

Let us consider a monitor $M$ and a node $N$.
In both cases (3) and (4), if the event occurs before executing $PeeV(N)$, no mismatch is produced in the local snapshot $G_M$.
Similarly, if a connection $N{\to}P$ is closed during a PeeV round, $M$ will not receive the \verb|marker| message and correctly remove $(N,P)$ from $E_M$.

Hence, only two cases can produce an error in the snapshot $G_M$:
when a connection $N{\to}P$ is dropped after $PeeV(N)$ completes, and 
when a connection $N{\to}P$ is established after $PeeV(N)$ has started.
The first case generates a false positive: $(N,P){\in}E_G$ but $(N,P){\notin}E$.
The second case generates a false negative: $(N,P){\in}E_G$ but $(N,P){\notin}E_M$.
Relatively to the local snapshot $G_M$, both mismatches will be fixed at the following PeeV round for $N$.
In particular, for any node $N$ and monitor $M$, errors can only exist during the time frame between two consecutive executions of $PeeV(N)$, which depends on the scan frequency $f_N^M$ set by $M$ for the target $N$.

Since the global snapshot is the majoritarian union of the local snapshots (a connection is included only if confirmed by the majority of monitors), 
a mismatch in $G_{AToM}$ can only last while half of the monitors contemporary have the error in their local snapshot.
Specifically, the time frame during which an error can stay in the global snapshot is, in the worst case, the smallest frequency for $N$ that is greater than the smallest $\frac{|\Gamma|}{2}$ frequencies.

Formally, the following theorem holds:
\begin{theorem}
	Let $[f_N^{1},f_N^{2},...,f_N^{|\Gamma|}]$ be the ordered list of all frequencies set by monitors $M_1,M_2,...,M_{|\Gamma|}$ for node $N$: .
	The longest time frame during which an error can exist in $G_{AToM}$ is $f_N^{\frac{|\Gamma|}{2}+1}$.
\end{theorem}
\begin{proof}
	Let $t_0$ be the time in which a change occurs.
	We consider the worst-case scenario where all monitors execute $PeeV(N)$ at time $t_0-1$.
	For the sake of simplicity, we assume the execution of $PeeV(N)$ as an atomic event.
	
	At time $t_0+f_N^1$ the first monitor detects the change. 
	At time $t_0+f_N^{\frac{|\Gamma|}{2}}$, half of the monitors will have detected the change.
	At time $t_0+f_N^{\frac{|\Gamma|}{2}+1}$, the majority of monitors will have detected the change, which is then reflected into $G_{AToM}$.
	
\end{proof}

Given the above, the accuracy of the global snapshot can be adjusted by setting a maximum frequency value for all monitors.

\subsection{Overhead}
In this section, we analyze the impact of AToM over network nodes, in terms of the number of extra messages introduced by our protocol and comparing it to the average number of messages currently exchanged by a node of the Bitcoin network.

We calculate the average number of messages exchanged by a node $N$ in a \textit{complete PeeV round}, which includes the execution of a PeeV round for $N$ (to verify its outbound peers) as well as the PeeV rounds needed to confirm all inbound peers.
During a complete round, nodes exchange the following messages:
\begin{itemize}
	\item a \verb!marker! message from each monitor in $\Gamma$, which is forwarded to all outbound peers;
	\item a \verb!verified! message from each monitor in $\Gamma$;
	\item a \verb!marker! message from each inbound peer, which is forwarded to the corresponding monitor.
\end{itemize}
Therefore, in a complete PeeV round, the total number of messages exchanged by a single node is: 
$$Msg_{PeeV} {=} (|O| + 2{\cdot}|I| + 1){\cdot}|\Gamma|,$$ where $|O|$ and $|I|$ are the number of outbound and inbound peers, respectively, and $|\Gamma|$ is the number of monitors.

As, on average, Bitcoin nodes experience around one change per hour in their outbound connections~\cite{statoshi-peers}, monitors will need to run, for each node, approximately one PeeV round per hour.
Hence, each node will exchange, for each monitor, around $Msg_{PeeV}$ messages per hour.
As mentioned in Section \ref{sec:background}, connections for a node are usually limited to 8 outbound and 117 inbound.
This value can actually be diverse in real life, with most nodes never reaching the limit~\cite{decker2013information,delgado2018txprobe}, and few others exceeding this number~\cite{miller2015discovering}.
Therefore, we consider $|O|{=}8$ and $|I|{=}117$ as average values and estimate the overhead of AToM for a single node as approximately $|O| {+} 2{\cdot}|I| {+} 1 = 8 {+} 2{\cdot}117 {+} 1 = 243$ additional messages per hour, for each monitor.

If we run, for instance, 10 monitors, each node would exchange around 2430 extra messages per hour, which is less than 1 extra message per second.
Following the same reasoning, we could run as many as 50 monitors with only 3 extra messages per second for each node.
Considering the average number of messages exchanged by a Bitcoin node is around 50 per second \cite{statoshi-msg}, we can say the overhead introduced by the AToM protocol is very low.

Moreover, AToM can easily scale to larger networks.
In fact, the cost of running the AToM protocol for a monitor increases linearly in the number of nodes, with an average of only 10 extra messages per new node (1 marker to the node, 8 markers from its outbound peers, and 1 verification message), in each PeeV round.
On the other hand, the impact on a single node does not depend on the size of the network, but only on the  number of monitors and the number of peers (which is limited by the Bitcoin client).

\section{Experimental Results}
\label{sec:experiments}

To evaluate our solution, we implemented a proof of concept (PoC) and performed experiments in a simulated environment. 
In this section, we describe our implementation and show the results of our experiments.

\subsection{Implementation}
We give details about our PoC implementation of AToM and describe our simulation environment.

\paragraph{\textbf{AToM}}
We implemented AToM by modifying Bitcoin Core 0.20.
We limited the compatibility of the protocol to one node per IP address (i.e., we do not allow two nodes to run from the same IP).
This limitation is due to the fact that inbound peers are assigned a random port, making it impossible to distinguish two different nodes connecting from the same IP. 
Considering that virtually all known public nodes run from a different IP, we consider this a minor issue.

The scan frequency and timeout values have been chosen according to our local network delays.
The PeeV timeout was set to 1 second, which proved to be sufficient for a full PeeV round ($M{-}N{-}P{-}M$).
Similarly, $adjustFrequency()$ works as described in Algorithm \ref{alg:atom}, using seconds as the time unit ($f_N$ is increased by 1 second and decreased by $c$ seconds).
In particular, for each node $N$, the initial scan frequency $f_N$ was set to 5 seconds, while the minimum and maximum values were set to $f_{min}{=}1$ and $f_{max}{=}10$, respectively.
The actual PeeV round scheduling is randomized following a Poisson distribution over $f_N$.
This further spreads messages over time and makes it harder for an adversary to predict PeeV round timings.

For safety purposes, nodes enable the reputation system for a node only after 3 PeeV rounds have passed (since the connection was established) for each monitor.
This was necessary to avoid disconnecting inbound peers that are not being scanned at the same rate as the node itself.
After this safe period, nodes are immediately disconnected when not verified by the majority of peers.
Such peers are also banned, avoiding future connections to and from them.

\paragraph{\textbf{Adversary}}
We implemented malicious nodes, which consistently deviate from the protocol to hide and fake connections.
These nodes are able to recognize each other and cooperate to deceive the monitors.
In particular, malicious nodes conceal their connections by dropping all markers received from honest peers and fake connections by forwarding those received from a monitor to other malicious nodes.
When a malicious node receives a marker from a colluding peer, it forwards it to the monitor, producing a false positive in the local snapshot.
Note that such behavior represents the worst-case scenario for our protocol.

\subsection{Evaluation}
We performed three series of experiments with network variability values of 1, 5, and 10 seconds, which have been chosen to be equal to the lowest, initial, and highest scan frequency for each node ($f_N$).
In each series, we performed runs with different percentages of malicious nodes in the network, ranging from 0\% (i.e., all honest nodes) to 50\%.

Each simulation lasted 10 minutes, with monitors probed every 30 seconds.
The global snapshot $G_{AToM}$ was calculated as in Definition \ref{def:gatom} and compared to the ground-truth topology $G$.

For each experiment, we counted the number of True Positives (TP), False Negatives (FN), and False Positives (FP, and then evaluated the accuracy of AToM in terms of precision and recall, defined as follows:
$$Precision = \frac{|TP|}{|TP|+|FP|}$$
$$Recall = \frac{|TP|}{|TP|+|FN|}.$$
These metrics are typically used to evaluate topology-inferring techniques, allowing a direct comparison with our protocol.
However, it should be noted that such techniques are run from the adversarial perspective.
This means, on the one hand, that honest nodes do not participate to the protocol, and, on the other hand, that there are no malicious nodes trying to cheat.
Therefore, false negatives and false positives are caused by other factors compared to our setting, giving precision and recall values a different meaning for topology-inferring techniques.
This should be taken into account when comparing these techniques to our protocol.

\paragraph{\textbf{Simulation Environment}}
To perform the experiments, we implemented a private Bitcoin network (i.e., Regtest) using our modified client and running nodes in Docker containers and managing them through a script.

The network is composed of 50 nodes and 4 monitors (the number of malicious nodes is calculated as the percentage of the total number of nodes).
Outbound connections are opened from each monitor towards all nodes, as soon as they are created.
In turn, each node connects to 3 outbound peers, chosen uniformly at random.
Between two nodes, only one connection can be established, meaning that there are no mutual connections, nor multiple inbound connections from the same node.

To cope with the local scale of the simulation, we let changes in the network occur at a much faster rate than the real Bitcoin network.
In particular, we have network events occur at a target average frequency, referred to as \textit{network variability} (and denoted by $var$ in our results), whose value is set at the beginning of each experiment.
At each iteration a node is either added or removed, maintaining an average of 50 nodes throughout the experiment. The percentage of malicious nodes is also kept stable over time.
When a node is removed, its inbound peers are connected to another node to always keep 3 outbound connections (this emulates the behavior of nodes in the Bitcoin network).

Despite the small scale, our framework is designed to behave as close as possible to the real Bitcoin network.
As such, although experiments are needed at a larger scale, we are confident that the results we obtained fairly represent the characteristics of our protocol.

\paragraph{\textbf{Experimental results}}
We show the AToM precision and recall obtained in our simulation Figure \ref{fig:precision} and Figure \ref{fig:recall}, respectively (values are shown in Table \ref{tab:precision} and Table \ref{tab:recall}).
\begin{table}
	\parbox{.48\linewidth}{
		\centering
		\resizebox{0.48\columnwidth}{!}{%
			\begin{tabular}{|c|c|c|c|c|c|c|l|}
				\hline
				& \textbf{0\%}    & \textbf{5\%} & \textbf{10\%} & \textbf{20\%} & \textbf{30\%} & \textbf{40\%} & \textbf{50\%}   \\ \hline
				\textbf{var=10} & 100 & 100  & 98.9 & 96.9 & 79.9 & 62.9 & 47.6 \\ \hline
				\textbf{var=5}  & 100 & 100  & 97.2 & 95.7 & 66.1 & 50.8 & 47.6 \\ \hline
				\textbf{var=1}  & 100 & 99.2 & 95.6 & 84.2 & 64.8 & 51.7 & 51.3 \\ \hline
			\end{tabular}
		}
		\caption{AToM precision (values)}
		\label{tab:precision} 
	}
	\hfill
	\parbox{.48\linewidth}{
		\centering
		\resizebox{0.48\columnwidth}{!}{%
			\begin{tabular}{|c|c|c|c|c|c|c|l|}
				\hline
				& \textbf{0\%}    & \textbf{5\%} & \textbf{10\%} & \textbf{20\%} & \textbf{30\%} & \textbf{40\%} & \textbf{50\%}   \\ \hline
				\textbf{var=10} & 100  & 99.5 & 98.0 & 95.9 & 88.3 & 79.5 & 54.6 \\ \hline
				\textbf{var=5}  & 99.9 & 99.0 & 97.3 & 95.8 & 82.5 & 75.7 & 50.1 \\ \hline
				\textbf{var=1 } & 99.8 & 98.8 & 96.7 & 91.3 & 82.3 & 68.6 & 43.0 \\ \hline
			\end{tabular}
		}
		\caption{AToM recall (values)}
		\label{tab:recall}
	}
\end{table}
\begin{figure}
	\centering
	\begin{subfigure}{.5\textwidth}
		\centering
		\includegraphics[width=0.95\columnwidth]{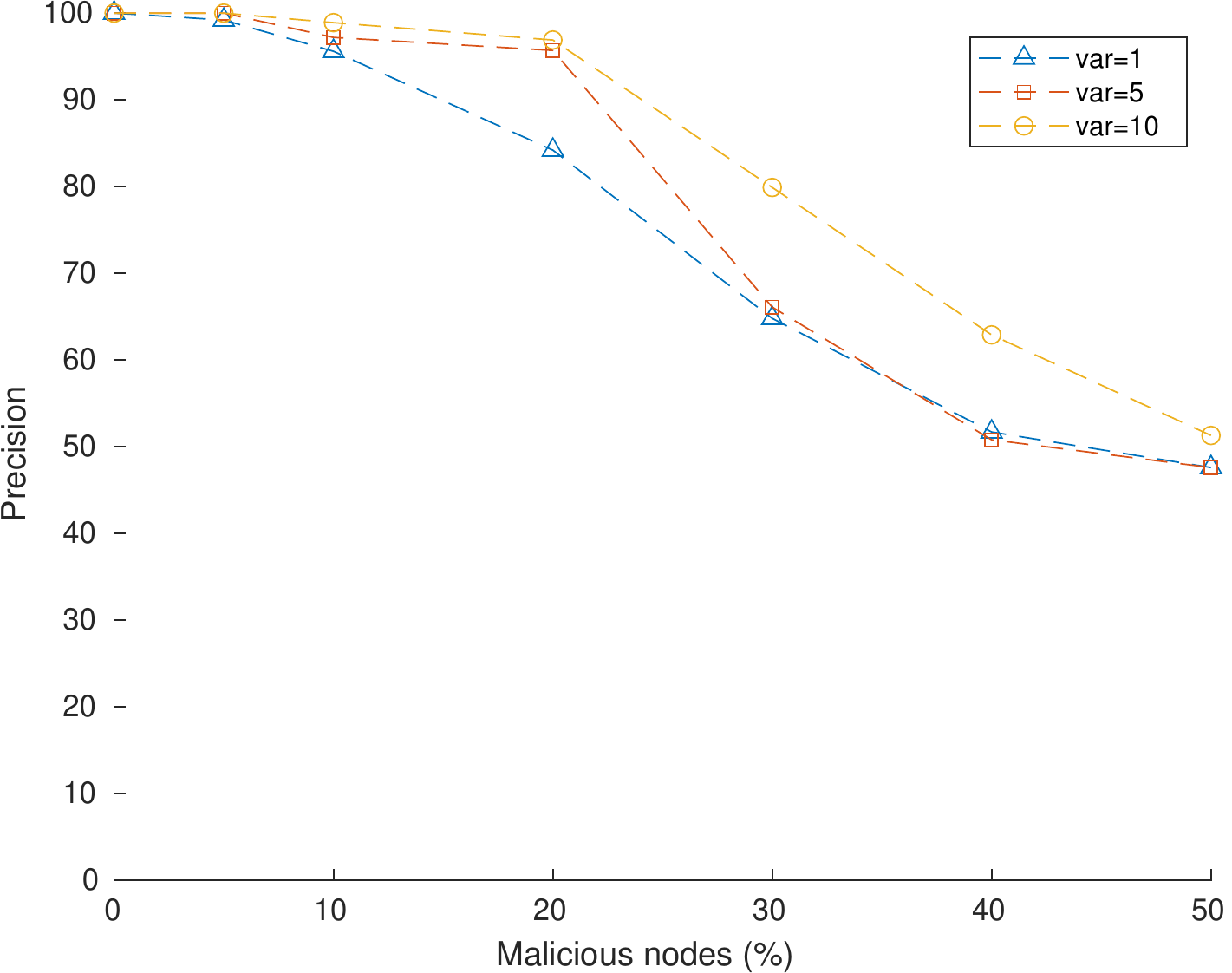}
		\caption{AToM precision.}
		\label{fig:precision}
	\end{subfigure}%
	\begin{subfigure}{.5\textwidth}
		\centering
		\includegraphics[width=0.95\columnwidth]{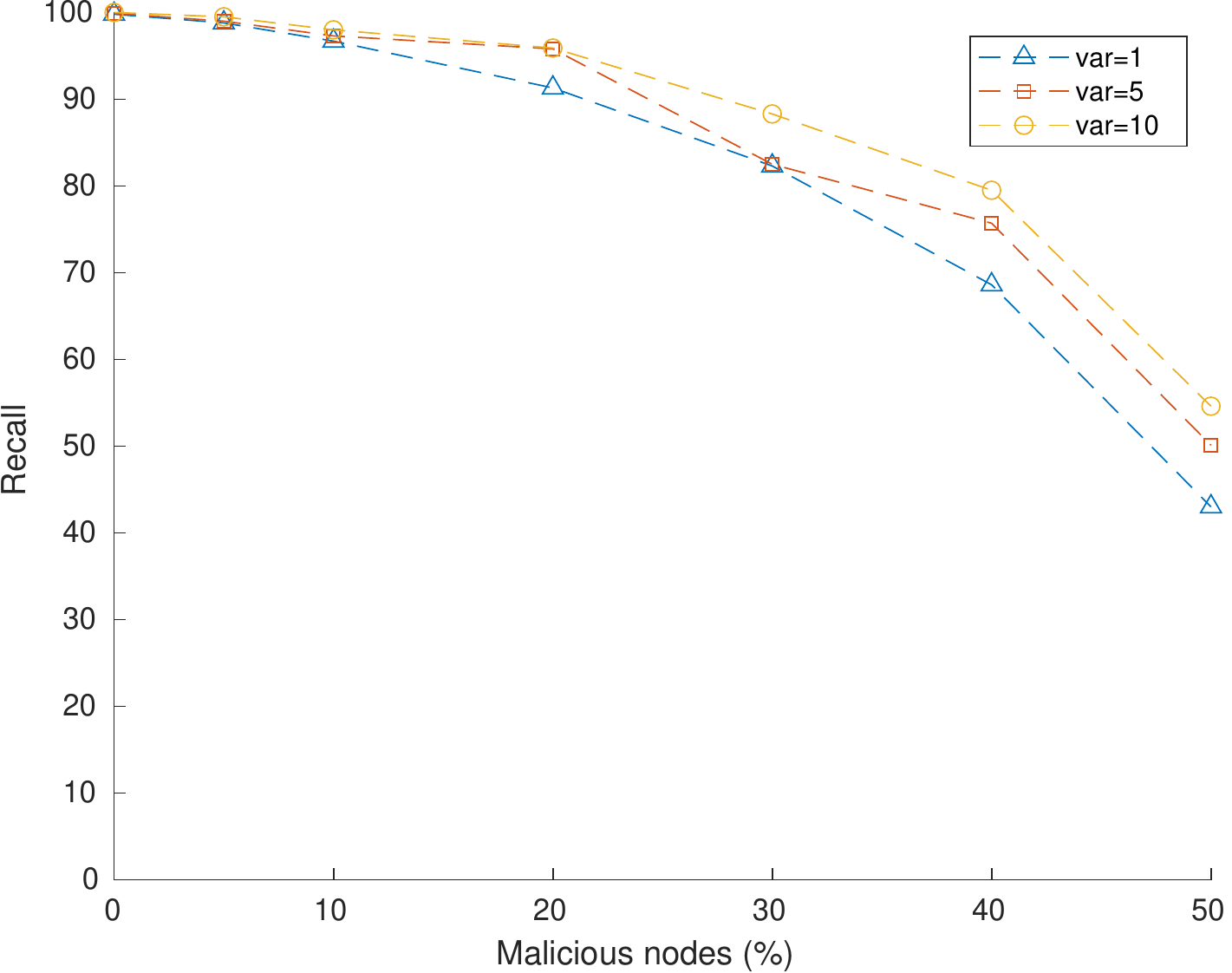}
		\caption{AToM recall.}
		\label{fig:recall}
	\end{subfigure}
	\caption{Experimental results for the AToM protocol.}
	\label{fig:test}
\vspace{-10pt}
\end{figure}

In an honest setting, both precision and recall are above 99\%, independently from the network variability.
In particular, the different network variability did not seem to produce notable differences in the results.

When malicious nodes are introduced, the number of false negatives and false positives starts to grow, lowering recall and precision, respectively.
However, recall decreases slower than precision, keeping over 90\% up to with 20\% malicious nodes and over 80\% with 30\% of malicious nodes.
This is due to the fact that colluding nodes, when connected, can only hide their own connection (thus generating a false negative), since the enforced consistency property quickly removes hidden connections with honest nodes.

On the other hand, as discussed in Section \ref{sec:design}, colluding nodes can exchange valid markers received from their inbound peers to make the monitor infer wrong connections.
This translates to as many false positives as their inbound peers.
For this reason, 
precision, while also keeping over 90\% with up to 20\% malicious nodes,
decreases more rapidly than recall when the number of colluding nodes grows.

Nevertheless, it is worth noting that controlling more than 20\% of the reachable network is very costly for the adversary, since it requires deploying thousands of nodes, located in different areas of the Internet (this is because Bitcoin clients limit the number of connections from a single subnet to reduce the risk of connecting to multiple adversarial nodes).
While controlling such a massive amount of nodes can be used to affect the accuracy of the global snapshot (but not, for instance, to cause disconnections among nodes)

Given the fact that controlling many reachable nodes can only affect the accuracy 
a similar attack could only affect the precision of the global snapshot, without being able, for instance, to cause disconnections among nodes, it is unclear whether such a motivation would be worth the cost for the adversary.

In conclusion, our results prove that AToM adapts well to the variability of the network and has high resilience against malicious nodes.
However, a massive number of colluding malicious nodes can affect its precision.
This factor should be taken into account when leveraging AToM for security purposes.

\section{Conclusion}
\label{sec:conclusion}
In this paper, we studied the problem of topology obfuscation in the Bitcoin network.
Despite the general agreement on the fact that the Bitcoin topology should be hidden for security reasons, no solid proofs have been provided in the literature to support this measure.

On the other hand, topology-inferring techniques are periodically disclosed, which allow attackers to effectively bypass this protective mechanism.
While most such techniques have been quickly made obsolete by fixes to the protocol, they show the difficulty of developers to enforce this limitation.
At the same time, researchers pointed out how the lack of topology information greatly hinders the analysis and optimization of the network, which are of paramount importance for its improvement.

We then addressed this problem by questioning the need of topology obfuscation and empirically investigating the security threats that are traditionally cited as a motivation for this approach.
We thoroughly reviewed all known network-level attacks related to or aided by knowledge of the topology.
Our analysis shows that most such attacks do not actually depend on this information or their risk has been mitigated over the years.

In light of this, we argued for an open topology, and proposed AToM, a semi-trusted monitoring system for the Bitcoin network.
Our protocol allows computing a full snapshot of the network and keeping it up-to-date over time.
We formally analyzed its correctness and showed its resilience against colluding malicious nodes.
We also experimentally evaluated its precision and recall in both trusted and untrusted settings.
Our solution has low overhead and can be implemented on any permissionless blockchain network, with virtually no modifications.

To the best of our knowledge, this is the first work to specifically address the debate on topology obfuscation and to propose a viable solution for the monitoring of the Bitcoin network.
Based on our findings, we endorse an open topology for the Bitcoin network and promote its active monitoring to help researchers and developers design a more efficient and secure cryptocurrency.

Future work includes a fully-distributed (trustless) version of the protocol, 
the application or introduction of topology-analysis techniques for the bitcoin network,
and the design of new mechanisms for improving connectivity and promptly react to security threats.

\section*{Acknowledgment}
This project is co-financed by the European Union Regional Development Fund within the framework of the ERDF Operational Program of Catalonia 2014-2020 with a grant of 50\% of total eligible cost.
\printbibliography
	
\end{document}